\documentclass[11pt]{article}

\usepackage{amsfonts}
\usepackage{enumitem}
\usepackage{mathtools}
\usepackage{soul}
\usepackage{booktabs}
\usepackage{latexsym}
\usepackage[ruled]{algorithm2e}
\SetAlFnt{\small}
\SetAlCapFnt{\small}
\SetAlCapNameFnt{\small}
\SetAlCapHSkip{0pt}
\IncMargin{-\parindent}
\usepackage{algorithmic}
\usepackage{nicefrac}

\usepackage[
  bookmarks=true,
  bookmarksnumbered=true,
  bookmarksopen=true,
  pdfborder={0 0 0},
  breaklinks=true,
  colorlinks=true,
  linkcolor=black,
  citecolor=black,
  filecolor=black,
  urlcolor=black,
]{hyperref}
\usepackage{graphicx}
\usepackage{xcolor}
\usepackage[margin=1in]{geometry}
\usepackage{amssymb}
\usepackage{amsmath}
\usepackage[square]{natbib}
\setcitestyle{numbers}
\usepackage{amsthm}
\usepackage[capitalize]{cleveref}
\usepackage{bbm}

\Crefname{figure}{Figure}{Figures}
\Crefname{equation}{Equation}{Equations}

\theoremstyle{plain}
\newtheorem{theorem}{Theorem}[section]
\newtheorem{corollary}[theorem]{Corollary}
\newtheorem{lemma}[theorem]{Lemma}
\newtheorem{proposition}[theorem]{Proposition}
\newtheorem{observation}[theorem]{Observation}
\newtheorem{remark}[theorem]{Remark}
\newtheorem{assumption}[theorem]{Assumption}
\newtheorem{fact}[theorem]{Fact}

\theoremstyle{definition}
\newtheorem{definition}[theorem]{Definition}
\newtheorem{notation}[theorem]{Notation}
\newtheorem{example}[theorem]{Example}

\DeclarePairedDelimiter\parentheses{(}{)}
\DeclarePairedDelimiter\braces{\{}{\}}
\DeclarePairedDelimiter\brackets{[}{]}
\DeclarePairedDelimiter\absolute{|}{|}

\DeclarePairedDelimiter\parhalf{(}{]}

\hypersetup{
  pdfauthor      = {Yakov Babichenko <yakovbab@technion.ac.il>, Inbal Talgam-Cohen <italgam@cs.technion.ac.il>, Konstantin Zabarnyi <konstzab@gmail.com>},
  pdftitle       = {Bayesian Persuasion under Ex Ante and Ex Post Constraints}
}
\begin{document}
\newcommand{\argmax}{{\mathrm{argmax}}}
\newcommand{\supp}{{\mathrm{supp}}}
\newcommand{\poly}{{\mathrm{poly}}}
\newcommand{\f}{{f}}
\newcommand{\g}{{g}}
\newcommand{\p}{{p}}
\newcommand{\HH}{{H}}
\newcommand{\LL}{{L}}

\title{Bayesian Persuasion under Ex Ante and Ex Post Constraints\thanks{This research has been supported by The Israel Science Foundation (grant \#336/18). The first author's research has been partially supported by The U.S.--Israel Binational Science Foundation (grant \#BSF 2026924) and by The German--Israeli Foundation for Scientific Research and Development (grant \#GIF 2027111); the second author is a Taub Fellow (supported by The Taub Family Foundation). The authors thank Ruggiero Cavallo for helpful conversations that motivated this research and anonymous reviewers for their helpful suggestions on improving this~paper.}}
\date{April 11, 2021}

\author{Yakov Babichenko\thanks{Technion--Israel Institute of Technology | \emph{E-mail}: \href{mailto:yakovbab@technion.ac.il}{yakovbab@technion.ac.il}.} \and Inbal Talgam-Cohen\thanks{Technion--Israel Institute of Technology | \emph{E-mail}: \href{mailto:italgam@cs.technion.ac.il}{italgam@cs.technion.ac.il}.} \and Konstantin Zabarnyi\thanks{Technion--Israel Institute of Technology | \emph{E-mail}: \href{mailto:konstzab@gmail.com}{konstzab@gmail.com}.}}

\maketitle

\begin{abstract}
Bayesian persuasion is the study of information sharing policies among strategic agents. A prime example is signaling in online ad auctions: what information should a platform signal to an advertiser regarding a user when selling the opportunity to advertise to her? Practical considerations such as preventing discrimination, protecting privacy or acknowledging limited attention of the information receiver impose constraints on information sharing. In this work, we propose and analyze a simple way to mathematically model such constraints as restrictions on Receiver's admissible posterior beliefs.

We consider two families of constraints -- ex ante and ex post, where the latter limits each instance of Sender-Receiver communication, while the former more general family can also pose restrictions in expectation. For the ex ante family, Doval and Skreta establish the existence of an optimal signaling scheme with a small number of signals -- at most the number of constraints plus the number of states of nature; we show this result is tight and provide an alternative proof for it. For the ex post family, we tighten a bound of Vølund, showing that the required number of signals is at most the number of states of nature, as in the original Kamenica-Gentzkow setting. As our main algorithmic result, we provide an additive bi-criteria FPTAS for an optimal constrained signaling scheme assuming a constant number of states; we improve the approximation to single-criteria under a Slater-like regularity condition. The FPTAS holds under standard assumptions; relaxed assumptions yield a PTAS. Finally, we bound the ratio between Sender's optimal utility under convex ex ante constraints and the corresponding ex post constraints. This bound applies to finding an approximately welfare-maximizing constrained signaling scheme in ad auctions.
\end{abstract}

\section{Introduction}
\label{sec:intro}
In many real-life situations, one entity relies on information revealed by another entity to decide which action to take. Call the former and the latter entities \emph{Receiver} and \emph{Sender}, respectively. Sender has the power to commit to a revelation policy, a.k.a.~a \emph{signaling scheme}. Sender would like to strategically design such a scheme to \emph{persuade} Receiver to act in Sender's interest. Mathematically, a signaling scheme transforms Receiver's prior belief about how some unknown \emph{state of nature} is distributed into a posterior belief, which determines Receiver's action.
 
Since strategic communication of information is intrinsic to most human endeavours, persuasion is of high importance in practice, and is becoming even more so in today's digital economy. Indeed, persuasion has been estimated to account for at least 30\% of the total US economy~\cite{MK,Antioch}. Persuasion has also attracted significant research interest in recent years, initiated by the celebrated Bayesian persuasion model of~\citeauthor{KG}.

\subsection{Our Contribution}
\label{sub:contribution}
We study a theoretical model for constrained Bayesian persuasion under general families of \emph{ex ante} and \emph{ex post} constraints. Ex ante constraints are statistical limitations on the amount of information Receiver may learn when the Sender-Receiver communication is repeated over time; ex post constraints are a strong particular case restricting the information passage on \emph{every} instance of the communication. These constraint families have various significant applications. In particular,~\citet{Tsakas} model signaling via noisy channels by ex ante-constrained persuasion.~\citet{DS} further show that optimal signaling via a capacity-constrained channel is equivalent to a constrained persuasion setting with a single entropy ex ante constraint.~\citet{Volund}, based on research in cognitive science, suggests ex post constraints as a possible model for human behaviour upon receiving an unwelcome signal. One of the main motivating examples in this work is online ad auctions in which ex ante constraints reduce discrimination and ex post constraints protect user privacy.

{\bf Our results and paper organization.} Let $m$ and $k$ be the numbers of constraints and states of nature, respectively. Section~\ref{sec:model} formally defines our model and describes the main motivations. Section~\ref{sec:existence} shows a tight bound of $k$ on the support size of an optimal ex post-constrained signaling scheme, which is the same as in the original setting of~\citeauthor{KG}. For ex ante constraints, Section~\ref{sec:existence} proves tightness of the $k+m$ bound of~\citet{DS} on the support size and provides an alternative proof to this bound; in particular, it extends the lower bound result of~\citet{TT} beyond a single constraint. The support size of a signaling scheme is a common measure of its complexity, similar to menu-size complexity in auctions~\cite{HN,HN17}. Section \ref{sec:computational} provides an additive bi-criteria FPTAS for an optimal signaling scheme when $k$ is constant and improves it to single-criteria under a Slater-like regularity condition. This result holds for standard constraints -- including Kullback–Leibler (KL) divergence, entropy and norm constraints (such as variation distance) -- and standard objective functions: Lipschitz-continuous (corresponding to Receiver having a continuum of actions) or piecewise constant (for finite Receiver's action space). Although these objective and constraint families capture a wide range of scenarios, the same algorithm remains an additive bi-criteria PTAS -- which improves to single-criteria under a Slater-like condition -- for even more general families. Section~\ref{sec:vs} shows that for constant $m$, convex constraints and a wide family of objective functions, ex ante constraints outperform ex post constraints by a constant multiplicative factor. Subsection~\ref{sub:app} concludes by applications to ad auctions with exponentially large states of nature space, using a generalization of the setting of~\citeauthor{BBX}~\cite{BBX}.

{\bf Technical challenges.} Ex ante constraints raise technical challenges not usually encountered in the literature on persuasion.  
In our model, we cannot restrict attention to \emph{straightforward policies}~\cite{KG} in which Sender recommends an action to Receiver in an incentive-compatible way. These policies are a very central tool in persuasion problems and are widely applied across the literature~\cite[see, e.g.,][]{Dughmi}, but they are not descriptive enough for determining whether a given ex ante constraint is satisfied. 
In particular, an optimal signaling scheme in our model cannot be described by a finite linear program (LP). Note that we do not assume Receiver's action space is finite, but even such a simplifying assumption would not have resolved these issues.

\subsection{Related Work}
\label{sub:related}
The seminal work of~\citet{KG} introduces Bayesian persuasion and characterizes Sender's optimal signaling scheme using the \emph{concavification} approach. Among the works on algorithmic aspects of persuasion we mention a negative result of~\citet{DX}, which is relevant to hardness of approximating the Sender's optimal utility; see~\cite{Dughmi} for a comprehensive survey of computational results.

In the context of auctions, an early work on signaling information is the classic paper of~\citet{MW}. \citet{EFGPT,MS} apply a computational approach to signaling in auctions;~\citet{FJMNTV} study signaling in the revenue-maximizing Myerson auction~\cite{Myerson};~\citet{BBX} study it in the welfare-maximizing second-price auction with exponentially many states of nature; and~\citet{DPT} design the signaling and auction mechanisms simultaneously.

The most closely related works to our own are the following: (a)~Our algorithmic approach in Section~\ref{sec:computational} is related to that of
\citet{CCDEHT}, as both use discretization and linear programming to achieve an additive FPTAS. (b)~\citet{DIR,DIOT} study constrained persuasion, but their constraints are on the complexity of the Sender-Receiver communication as measured by message length or number of signaled features and so are fundamentally different from ours.%
\footnote{They also study a version called ``bipartite signaling'', which has a combinatorial flavour different than ours, in an auction setting with the strong assumption that bidder values are known.} \citet{Ichihashi} considers persuasion by Sender who is constrained in the information she can acquire (and therefore, send) and characterizes the set of possible equilibrium outcomes. Our Theorem~\ref{T1} is related to this literature in that it indicates that ex post constraints on persuasion do not cause a blowup in the number of signals needed to persuade optimally. (c)~Inspired by~\cite{TT},~\citet{DS} prove an upper bound on the required number of signals in an ex ante-constrained optimal scheme; we show that this bound is tight, give an alternative proof to the bound and provide an analogous tight bound for ex post constraints in Section~\ref{sec:existence}. (d)~\citet{Volund} studies a model of persuasion on compact subsets, which is equivalent to our ex post constraints; there is no parallel in that work to ex ante constraints, and the results on ex post in the two works do not overlap.

In Subsection~\ref{sub:motivation}, we discuss motivating applications of ex ante and ex post constraints, including limited attention, as well as privacy protection in online ad auctions. \cite{BS,LMW} study persuasion with limited attention -- see Subsection~\ref{sub:motivation} for details.~\citet{EEM} study ex ante and ex post privacy constraints in the design of auctions rather than persuasion schemes. \citet{Ichihashi20} studies the economic implications of online consumer privacy; in his model, the consumer, rather than the seller, plays the role of Sender. It is important to note that the differential privacy paradigm~\cite[see][]{DR} does \emph{not} apply to privacy protection in online ad auctions: the state of nature about which information is revealed represents characteristics of an individual rather than statistics of a large population, and it is inherent to ad personalization that these characteristics influence the outcome in a non-negligible way.

\section{Our Model}
\label{sec:model}
\subsection{Bayesian Persuasion Preliminaries}
\label{sub:preliminaries}
We consider Bayesian persuasion with a single \emph{Sender} and a single \emph{Receiver}, as introduced by~\citet{KG}. Fix a space of $k$ \emph{states of nature} $\Omega$ and a commonly-known \emph{prior distribution} $\p$ on them. Take some compact nonempty set $A$ to be Receiver's \emph{action space}. Introduce two random variables $\omega$ and $x$, representing the \emph{state of nature} and Receiver's \emph{action}, respectively. Fix a \emph{Sender's utility function} $\tilde{u}_s:A\times\Omega\to \mathbb{R}_{\geq 0}$ and a \emph{Receiver's utility function} $u_r:A\times \Omega \to \mathbb{R}_{\geq 0}$. The Sender-Receiver communication is specified by a \emph{signaling scheme} $\Sigma$, a.k.a.~a \emph{signaling policy}, which is a randomized function from $\Omega$ to some set of signals (this notion will be formalized soon). Sender must commit to $\Sigma$ before learning $\omega$.

Denote by $\Delta\parentheses*{\Omega}$ the set of probability distributions over $\Omega$. Consider it to be a subset of $\brackets*{0,1}^k$, with $i$-th coordinate being the probability assigned to the $i$-th element of $\Omega$.

Let $\sigma$ be the actual \emph{signal realization}. Note that $\sigma$ induces an updated distribution on $\Omega$ in Receiver's view, called the \emph{posterior distribution} or the \emph{posterior}. Let $\p_{\sigma}\in\Delta\parentheses*{\Omega}$ be the posterior induced by $\sigma$. The \emph{support} of $\Sigma$, $\supp\parentheses*{\Sigma}$, is the intersection of all the closed sets $S\subseteq\Delta\parentheses*{\Omega}$ s.t.~$\Pr_{\Sigma}\brackets*{\p_{\sigma}\in S}=1$. If $\Sigma$ uses only countably many signals, then $\supp\parentheses*{\Sigma}$ is the set of all the posteriors induced by signal realizations of $\Sigma$ with a positive probability.

Formally, $\Sigma$ is a distribution, unconditional on the state of nature, over the elements of $\Delta\parentheses*{\Omega}$ that belong to $\supp\parentheses*{\Sigma}$. For any $\omega_0\in\Omega$, assuming $\omega=\omega_0$, $\Sigma$ induces a conditional distribution over the elements of $\supp\parentheses*{\Sigma}$ that specifies how Sender chooses the signal realization when $\omega=\omega_0$. Denote this distribution by $\Sigma\parentheses*{\omega_0}$. Note that given $p$ and $\Sigma$, it can be computed by Bayes' law.

For simplicity, we introduce the following notation for the expectation of a function of the posterior over the elements of $\supp\parentheses*{\Sigma}$ according to $\Sigma$:

\begin{notation}
\label{N1}
For a function $\f:\Delta\parentheses*{\Omega}\to\mathbb{R}$:
\begin{equation*}
E\brackets*{\Sigma,\f}:=\mathbb{E}_{\p_{\sigma}\sim \Sigma}\brackets*{\f\parentheses*{\p_{\sigma}}}=\mathbb{E}_{\omega\sim\p, \p_{\sigma}\sim \Sigma\parentheses*{\omega}}\brackets*{\f\parentheses*{\p_{\sigma}}}.
\end{equation*}
\end{notation}

By~\cite{Blackwell,AMS}, a distribution $\Sigma$ represents a signaling scheme if and only if $\Sigma$ is \emph{Bayes-plausible}; that~is:

\begin{equation*}
\label{EQ1}
\forall \omega_0 \in\Omega:\;\; \p\brackets*{\omega_0}=E\brackets*{\Sigma,\p_{\sigma}\brackets*{\omega_0}}.
\end{equation*}

The persuasion process runs as follows:
(1)~Sender commits to a signaling policy $\Sigma$.
(2)~Sender discovers the state of nature $\omega$.
(3)~Sender transmits a signal realization $\sigma$ to Receiver, according to $\Sigma\parentheses*{\omega}$.
(4)~Receiver chooses an action $x\in A$ s.t.~$x\in\argmax\parentheses*{\mathbb{E}_{\omega'\sim \p_{\sigma}}\brackets*{u_r \parentheses*{x,\omega'}}}$; assume, as is standard, that ties are broken in Sender's favour.
(5)~Sender gets utility of $\tilde{u}_s \parentheses*{x,\omega}$, while Receiver gets utility of $u_r \parentheses*{x,\omega}$.

Since $x$ depends only on $\p_{\sigma}$, there exists $\bar{u}_s:\Delta\parentheses*{\Omega}\times\Omega\to \mathbb{R}_{\geq 0}$ s.t.~$\tilde{u}_s\parentheses*{x,\omega}\equiv \bar{u}_s\parentheses*{\p_{\sigma},\omega}$. Define $u_s:\Delta\parentheses*{\Omega}\to \mathbb{R}_{\geq 0}$ by $u_s\parentheses*{p_{\sigma}}:=\mathbb{E}_{\omega'\sim\p_{\sigma}}\brackets*{\bar{u}_s\parentheses*{\p_{\sigma}, \omega'}}$.

\begin{remark}
\label{R1}
From now on we shall consider $u_s$ instead of $\tilde{u}_s$ or $\bar{u}_s$, assuming, therefore, that Sender's utility is state of nature-independent. This is w.l.o.g.~for our theorems from Sections~\ref{sec:existence}-\ref{sec:computational}, since the passage from $\bar{u}_s$ to $u_s$ preserves the conditions required there (being upper semi-continuous, continuous, piecewise constant or $O(1)$-Lipschitz).\footnote{In the state-dependent setting, $\bar{u}_s\parentheses*{\cdot,\omega_0}$ has to satisfy the theorem requirements from $u_s$ for every $\omega_0\in\Omega$.} 
While one cannot apply the results of Section~\ref{sec:vs} to the state-dependent case without strengthening Assumption~\ref{A1}, the natural applications to ad auctions discussed there have state-independent Sender's utility.
\end{remark}

Throughout we make the following assumption, which is a relaxation of the standard assumption in the persuasion literature that $u_s$ is continuous. In particular, this assumption encompasses $u_s$ that is a threshold function.

\begin{assumption}
\label{A8}
The function $u_s$ is upper semi-continuous.
\end{assumption}

\subsection{Ex Ante and Ex Post Constraints}
\label{sub:constraints}
So far we have described the setting of~\citet{KG}. However, in our model we do not allow Sender to choose among all Bayes-plausible signaling schemes, but only among schemes that satisfy certain restrictions (see Subsection~\ref{sub:motivation} for motivation). We define two general families of constraints: \emph{ex ante} and \emph{ex post}. A constraint of the latter type restricts the admissible values of a certain function of $p_{\sigma}$ for every possible $p_{\sigma}$, while a constraint of the former type restricts only the expectation of such a function.

\begin{definition}[Ex ante constraints]
\label{D1}
An \emph{ex ante constraint} on a signaling scheme $\Sigma$ is a constraint of the form:
\begin{equation*}
E\brackets*{\Sigma,\f}\leq c
\end{equation*}
for continuous $f:\Delta\parentheses*{\Omega}\to\mathbb{R}$ and a constant $c\in\mathbb{R}$.
\end{definition}

\begin{definition}[Ex post constraints]
\label{D5}
An \emph{ex post constraint} on a signaling scheme $\Sigma$ is a constraint of the form:
\begin{equation*}
\forall\p_{\sigma}\in\supp\parentheses*{\Sigma}:\;\\f\parentheses*{\p_{\sigma}}\leq c
\end{equation*}
for continuous $f:\Delta\parentheses*{\Omega}\to\mathbb{R}$ and a constant $c\in\mathbb{R}$.
\end{definition}

For a constraint defined as in either of the previous two definitions, we say that the constraint is \emph{specified} by the function $\f$ and the constant $c$. A constraint specified by a convex $f$ and some constant $c$ is called \emph{convex}.

\begin{observation}
\label{O1}
Ex post constraints are a special case of ex ante constraints.
\end{observation}

Indeed, an ex post constraint specified by some $\f$ and $c$ is equivalent to the ex ante constraint specified by $\max\braces*{\f,c}$ and $c$. Note that if $f$ is convex then so is $\max\braces*{\f,c}$.

Every ex ante constraint can be transformed into a (stronger) ex post constraint by "erasing the expectation" and vice versa. Formally:

\begin{definition}
\label{D2}
An ex post and an ex ante constraint \emph{correspond} to each other if they are specified by the same function and the same constant.
\end{definition}

\begin{definition}
\label{D3}
Given a set of constraints, a signaling scheme satisfying all of them is called \emph{valid}.
\end{definition}

\begin{definition}
\label{D4}
A set of constraints is called \emph{trivial} if every signaling scheme satisfies it.
\end{definition}

\subsection{Motivation for Constrained Persuasion}
\label{sub:motivation}
In many applications of~\citeauthor{KG}'s model, Sender may not be able to reveal as much information as would theoretically be optimal due to imposed constraints. Such constraints can originate from sources including law, professional integrity, political agreements, public opinion and limited attention.

{\bf Online ad auctions.} In this first motivating example, the auctioneer -- an advertising platform -- is Sender, while the set of bidders -- which are advertisers -- is Receiver.%
\footnote{We treat the bidders as a single Receiver since they all get the same signal; private signaling poses additional challenges~\cite{AB} and is left for future work.}
The profile of the web user who is about to view the ad is the state of nature. This profile is known to the auctioneer, but not to the bidders; every signal reveals information about it. Such information revelation should be restricted by both privacy and fairness considerations. 

The constraint families we introduce are suitable for protecting privacy: following~\citet{EEM}, privacy protection can be modeled as imposing a threshold on the KL divergence from the prior to the posterior. The KL divergence quantifies how much more informative the posterior is compared to the prior due to extra information about the user provided by the signal realization. On the one hand, an ex post constraint on the KL divergence provides a relatively robust protection of individual privacy by ruling out sending a very informative signal even with only a small probability. On the other hand, the corresponding ex ante constraint protects privacy on the group level -- e.g., it limits Receiver's ability to learn the shopping habits of certain population groups, since the posterior is close, on average, to the prior.

Another important restriction on signaling in ad auctions is fairness, or anti-discrimination -- e.g., ensuring that enough women compared to men are shown an ad for a high-paying job~\cite{CKSV,CMV}. Consider, for simplicity, a uniform prior over population groups. A simple constraint specified by $(-\min_{\omega'\in\Omega}\{ p_{\sigma}\brackets*{\omega'}\})$ lower bounds the frequency of a population group in the posterior, ensuring, therefore, its proportional inclusion.%
\footnote{
If the prior over population groups is not uniform, then we can easily add weights to this constraint: $-\min_{\omega'\in\Omega}\braces*{b_{\omega'} p_{\sigma}\brackets*{\omega'}}$.
}
An ex ante constraint of this form ensures that on average, the advertiser does not get enough information to discriminate against particular groups.

{\bf Limited attention.} A second motivating example involves constraints arising from Receiver's limited attention span. As Simon (\citeyear{Simon}) noted, ``a wealth of information creates a poverty of attention''. Our model enables limiting the signaled information so that it ``fits'' within Receiver's limited attention.%
\footnote{An alternative model of~\cite{BS,LMW} allows Sender to ``flood'' Receiver with information, but Receiver strategically chooses what to pay attention to. 
Constrained persuasion might be viewed as a restriction that simply avoids flooding Receiver with information in expectation (the ex ante model) or always (the ex post model).
}
Following the \emph{rational inattention} literature~\cite{Sims}, define the attention required from Receiver to process Sender's signal $\sigma$ as the entropy of the posterior $p_\sigma$.%
\footnote{\citet{BS} use mutual information of $p_\sigma$ and Receiver's perception of it after paying limited attention as the measure of the attention invested by Receiver. In our model, Receiver always pays full attention; thus, the mutual information coincides with the entropy of $p_\sigma$.}
By constraining the entropy -- either in expectation (i.e., ex ante) or of every posterior (i.e., ex post) -- we enable Receiver to process the signal despite her limited attention (where the limit is either in expectation or per signal, respectively). A concrete application from~\citet{BS} includes a busy executive as Receiver, one of her advisors as Sender and constraints on the signaled information enforced by keeping meetings and briefings short (on average or per meeting). 

\section{Existence Results}
\label{sec:existence}
\citet{DS}~prove that for every set of $m$ ex ante constraints, there exists an optimal valid signaling scheme with support size of at most $k+m$:
    \begin{fact}[\citet{DS} -- existence of an optimal valid signaling scheme under ex ante constraints with a linear-sized support]
    \label{F1}
    Fix $m$ ex ante constraints. Then either there exists an optimal valid signaling scheme with support size at most $k+m$ or the set of valid signaling schemes is empty.
    \end{fact}
We show that this bound is tight. We further prove that for any number $m$ of ex post constraints, a stronger tight bound of $k$ holds, just as in the unconstrained setting of~\citeauthor{KG}; the same proof outline yields an alternative proof to the result of~\cite{DS} on ex ante constraints, as shown in Appendix~\ref{app:existence_alternative}.
\begin{proposition}
    \label{P3}
    The bound from Fact~\ref{F1} on the support size is tight for every $k$ and $m$.
    \end{proposition}
    We provide a constructive proof to Proposition~\ref{P3} in Appendix~\ref{app:existence}. Now we establish a stronger bound for ex post constraints.
    \begin{theorem}[Existence of an optimal valid signaling scheme under ex post constraints with a linear-sized support]
    \label{T1}
    Fix a set of ex post constraints. Then either there exists an optimal valid signaling scheme with support size at most $k$ or the set of valid signaling schemes is empty.
    \end{theorem}
     At a high level, we translate the problem into an infinite LP, with the ``variables'' being the distribution $\Sigma$ over $\Delta\parentheses*{\Omega}$. We first prove that the target function of the infinite LP is upper semi-continuous. Secondly, we show, using infinite-dimensional optimization tools, that it must attain a maximum at an extreme point of the feasible set. Thirdly, we argue that every extreme point has a finite support of bounded size, analyzing the effect of adding the Bayes-plausibility constraints one by one by considering the hyperplanes specifying the constraints: the maximal support size of extreme points is at most doubled upon each addition. Finally, we improve the bound on the support size of each extreme point using a finite LP.\footnote{If the ex post constraints are convex, the result follows directly from the concavification approach of~\citet{KG}.}
\begin{proof}[Proof of Theorem~\ref{T1}]
Denote the function and the constant specifying the $i$-th ex post constraint ($1\leq i\leq m$) by $\f_i$ and $c_i$, respectively. Let $K:=\cap_{1\leq i\leq m} f_i^{-1}\parhalf*{-\infty, c_i}\subseteq \Delta\parentheses*{\Omega}$ be the set of posteriors that are allowed to belong to a support of a valid signaling scheme. As $f_1,...,f_m$ are continuous, $K$ is compact. 
    
    We would like to solve, assuming $\supp\parentheses*{\Sigma}\subseteq K$:
    \begin{align*}
    \text{max}\;\;\;\;&E\brackets*{\Sigma,u_s}\\
    \text{s.t.}\;\;\;\;\;\;&\p\brackets*{\omega_0}=E\brackets*{\Sigma,\p_{\sigma}\brackets*{\omega_0}}\;\;\forall \omega_0 \in\Omega
    \end{align*}
    The above optimization problem is an infinite LP, with the ``variables'' being the distribution $\Sigma$ over $K$. Consider the metric space of the feasible probability measures on $K$ with the Lévy–Prokhorov metric. Take a sequence $\braces*{\mu_n}_{n\geq 1}$ of feasible probability measures on $K$ that converges to a feasible measure $\mu$. Since $K$, equipped with the Euclidean metric, is separable, we get from a well-known result (e.g., Theorem~4.2 from~\cite{VanGaans}) that $\mu_n$ weakly converges to $\mu$.\footnote{The well-known result states that for a separable metric space $\parentheses*{X,d}$, convergence of measures on it in the Lévy–Prokhorov metric and weak convergence of measures are equivalent.} $u_s$ is upper semi-continuous and defined on a compact set; thus, it is bounded from above. From one of the equivalent definitions of weak convergence of measures:
   \begin{equation*}
   \limsup E\brackets*{\mu_n, u_s}\leq E\brackets*{\mu, u_s}.
   \end{equation*}
   Therefore, the target function in the infinite LP is upper semi-continuous with respect to the Lévy–Prokhorov metric on the space of the feasible probability measures and the usual metric on $\mathbb{R}_{\geq 0}$. This completes our first step.

    The target function is upper semi-continuous and linear, and the feasible set of measures is compact and convex; thus, Bauer's maximum principle (e.g., Theorem~7.69 from~\cite{Bauer}) yields that an optimum is attained at an extreme point (unless the feasible set is empty and no valid signaling scheme exists), which completes our second step. It remains to show that every extreme point of the feasible set has support of size at most $k$.
    
    A general approach adapted from~\cite{Richter,MR,Karr} shows that every extreme point has a finite support with size at most $2^{k}$. This is because every constraint in the infinite LP is defined by a hyperplane; when adding the hyperplanes one by one -- the maximal support size of extreme points is at most doubled upon each addition. It completes our third step.
    
    Finally, discretize our LP by setting $\absolute*{\supp\parentheses*{\Sigma}}\leq 2^{k}$ and considering each of the infinitely many candidates for $\supp\parentheses*{\Sigma}$ separately. Each candidate defines a finite LP with $2^{k}$ variables and $k$ constraints (we should add a constraint ensuring that the probability masses in $\Sigma$ sum up to $1$, but then one of the Bayes-plausibility constraints becomes redundant). Thus, every extreme point of the infinite LP -- which is an extreme point of some finite LP -- is supported on at most $k$ coordinates, which completes the proof.
\end{proof}
    
    \begin{observation}
    \label{O3}
    The bound from Theorem~\ref{T1} is achieved, e.g., by $u_s\parentheses*{p_{\sigma}}:=||p_{\sigma}||_{\infty}$ and a set of trivial ex post constraints.
    \end{observation}
    
\section{Computational Aspects}
\label{sec:computational}
In this section, we provide positive computational results for a constant number of states of nature $k$. We focus on constant $k$ since a hardness result of~\citet{DX} implies that unless $P=NP$, there is neither an additive PTAS nor a constant-factor multiplicative approximation of the optimal Sender's utility in $\poly(k)$-time, even for piecewise constant $u_s$.\footnote{Their result is on public persuasion with multiple Receivers, which can be replaced by a single Receiver with a large action space.} Our results are for ex ante constraints; by Observation~\ref{O1}, they hold also for ex post constraints. Throughout this section, we assume that both $u_s$ and the functions specifying the constraints are given by explicit formulae and can be evaluated at every point in constant time.

Call \emph{$L$-Lipschitz} a function with Lipschitz constant being at most $L$. Our first main result is an additive bi-criteria approximation (Theorem~\ref{T2}). Part $1$ of Theorem~\ref{T2} is an additive bi-criteria FPTAS for $O(1)$-Lipschitz or piecewise constant $u_s$ and a natural constraint family that includes entropy, KL divergence and norms. This result encompasses the utility functions that naturally arise in applications of Bayesian persuasion: piecewise constant if Receiver has finitely many actions and $O(1)$-Lipschitz if Receiver has a continuum of actions~\cite{Dughmi}. Specifically, we show how to compute in $\poly\parentheses*{m, \frac{1}{\epsilon}}$-time a signaling scheme achieving utility that is additively at most $\epsilon$-far from optimal and violating each of the $m$ ex ante constraints by at most~$\epsilon$; Bayes-plausibility is satisfied precisely. Part $2$ of Theorem~\ref{T2} is an additive bi-criteria PTAS, which holds under even weaker assumptions: $u_s$ should be either continuous or piecewise constant and there are no limitations on the ex ante constraints. The same approximation algorithm implies both parts of Theorem~\ref{T2}.

Our second main result (Theorem~\ref{T3}) is an improvement of the bi-criteria approximations from Theorem~\ref{T2} to single-criteria; it requires imposing a Slater-like regularity condition on the ex ante constraints. We provide the main steps of the proofs of our computational results in Subsection~\ref{sub:proofs}, while the remaining details are in Appendix~\ref{app:computational}.

To make the theorem statements as general as we can, we have introduced some technical assumptions. To present some motivations for the general results and improve clarity, we first state two special cases of our main results. First, for every continuous $u_s$, there exists an additive bi-criteria PTAS for an optimal signaling scheme; secondly, the same algorithm is an additive bi-criteria FPTAS when both $u_s$ and the functions specifying the ex ante constraints are $O(1)$-Lipschitz; both results improve to single-criteria approximations under a Slater-like regularity condition.

\begin{corollary}[of Theorems~\ref{T2},~\ref{T3}]
\label{C1}
Suppose that $k$ is constant, $u_s$ is continuous and given are $m$ ex ante constraints s.t. the set of valid signaling schemes is nonempty. Then for every $\epsilon>0$, there exists an algorithm that computes an additively $\epsilon$-optimal signaling scheme that violates each ex ante constraint at most by $\epsilon$, which has running time of:
\begin{enumerate}
    \item $\poly\parentheses*{m}$, provided that $\epsilon$ is constant.
    \item $\poly\parentheses*{m, \frac{1}{\epsilon}}$, provided that both $u_s$ and the functions specifying the ex ante constraints are $O(1)$-Lipschitz.
\end{enumerate}
Furthermore, if there exists a signaling scheme satisfying each ex ante constraint with strict inequality, then the above algorithm can be improved so that each ex ante constraint is satisfied precisely.
\end{corollary}

\begin{remark}
\label{R2}
All the approximation algorithms from Section~\ref{sec:computational} output a solution of a finite LP with $k+m$ constraints; therefore, their output -- which is w.l.o.g.~a basic feasible solution -- is a signaling scheme with support size at most $k+m$, which matches the tight theoretical bound from Fact~\ref{F1}.
\end{remark}
\subsection{Bi-criteria Approximation}
\label{sub:bi}
Here we present an additive bi-criteria FPTAS (Theorem~\ref{T2}, part $1$) for $O(1)$-Lipschitz or piecewise constant Sender's utility functions, under ex ante constraints specified by functions which may include entropy, KL divergence and any norm of $p_{\sigma}-p$ (such as the well-known variation distance between probability measures).\footnote{Note that every norm on $\Delta\parentheses*{\Omega}\subseteq \mathbb{R}^k$ is $O(1)$-Lipschitz, which is sufficient to satisfy Assumption \ref{A6}.} In particular, one can restrict $D_{KL}\parentheses*{p_{\sigma}'||p'}$, where $p_{\sigma}'$ and $p'$ are the distributions induced by $p_{\sigma}$ and $p$ (respectively) on some partition of $\Omega$; that is, some elements of $\Omega$ are united when computing the KL divergence. Practically, it can be exploited in online ad auctions to limit the expected information disclosure on habits of a certain social group; such a group is represented by a subset of $\Omega$.

\begin{assumption}[$u_s$ is $O(1)$-Lipschitz or piecewise constant -- required for the additive bi-criteria FPTAS]
\label{A5}
$u_s$ is either $O(1)$-Lipschitz or piecewise constant, with a constant number of pieces, s.t.~each piece covers a convex polygon in $\Delta\parentheses*{\Omega}$ with a constant number of vertices.
\end{assumption}

\begin{assumption}[The ex ante constraints are specified by $O(1)$-Lipschitz functions, entropy or KL divergence -- required for the additive bi-criteria FPTAS]
\label{A6}
Each ex ante constraint is specified either by an $O(1)$-Lipschitz function or by a function of the form:
\begin{equation*}
b\cdot \sum_{1\leq j\leq l} \parentheses*{\sum_{\omega'\in\Omega_j} p_{\sigma}\brackets*{\omega'}}\ln\frac{\sum_{\omega'\in\Omega_j} p_{\sigma}\brackets*{\omega'}}{b_j},
\end{equation*}
where $\braces*{\Omega_j}_{1\leq j\leq l}$ is a partition of $\Omega$ and $b,b_1,...,b_l$ are constants ($b_1,...,b_l>0$).
\end{assumption}

We further show that under no assumptions on the ex ante constraints and under a weaker assumption on $u_s$ -- being continuous or piecewise constant -- the same algorithm provides an additive bi-criteria PTAS (Theorem~\ref{T2}, part $2$).

\begin{assumption}[$u_s$ is continuous or piecewise constant -- relaxation of Assumption \ref{A5}; required for the additive bi-criteria PTAS]
\label{A7}
$u_s$ is either continuous or piecewise constant, with a constant number of pieces, s.t.~each piece covers a convex polygon in $\Delta\parentheses*{\Omega}$ with a constant number of vertices.
\end{assumption}

\begin{theorem}[An additive bi-criteria FPTAS/PTAS for an optimal valid signaling scheme]
\label{T2}
Fix a constant $k$ and fix $m$ ex ante constraints s.t. the set of valid signaling schemes is nonempty.
\begin{enumerate}
    \item Suppose that $u_s$ satisfies Assumption~\ref{A5} and the ex ante constraints satisfy Assumption~\ref{A6}. Then for every $\epsilon>0$, there exists a $\poly\parentheses*{m, \frac{1}{\epsilon}}$-algorithm that computes an additively $\epsilon$-optimal signaling scheme that violates each ex ante constraint at most by $\epsilon$.
    \item Suppose that $u_s$ satisfies Assumption~\ref{A7}. Then for every constant $\epsilon>0$, there exists a $\poly\parentheses*{m}$-algorithm that computes an additively $\epsilon$-optimal signaling scheme that violates each ex ante constraint at most by $\epsilon$.
\end{enumerate}
\end{theorem}
\subsection{Single-criteria Approximation}
\label{sub:single}
So far, we have demonstrated how to find a near-optimal signaling scheme that satisfies the ex ante constraints after slightly relaxing them. The relaxation is required to avoid degenerate cases. For example, finding the root of a polynomial with a single real root can be described in the language of ex ante constraints. This problem has a unique feasible distribution and if we do not relax the constraints, any algorithm missing the exact real root cannot give a satisfactory approximation. Theorem~\ref{T2} can be improved under a regularity condition disallowing such degenerate cases.
\begin{assumption}[Slater-like regularity condition]
\label{A3}
There exists a signaling scheme satisfying all the given ex ante constraints with strict inequality.
\end{assumption}
\begin{theorem}[An additive FPTAS/PTAS for an optimal valid signaling scheme]
\label{T3}
Fix a constant $k$ and fix $m$ ex ante constraints satisfying Assumption~\ref{A3}.
\begin{enumerate}
    \item Suppose that $u_s$ satisfies Assumption~\ref{A5} and the ex ante constraints satisfy Assumption~\ref{A6}. Then for every $\epsilon>0$, there exists a $\poly\parentheses*{m, \frac{1}{\epsilon}}$-algorithm that computes an additively $\epsilon$-optimal valid signaling scheme.
    \item Fix a constant $\epsilon>0$ and suppose that $u_s$ satisfies Assumption~\ref{A7}. Then there exists a $\poly\parentheses*{m}$-algorithm that computes an additively $\epsilon$-optimal valid signaling scheme.
\end{enumerate}
\end{theorem}
\subsection{Proofs of the Computational Results}
\label{sub:proofs}
In this subsection, we first formulate and prove Lemma~\ref{L3} (together with two technical assumptions), which is the main step in the proofs of our results from Section \ref{sec:computational}. The first and second parts of Theorem~\ref{T2} follow from this lemma, with $t\parentheses*{\frac{1}{\epsilon}}:=\frac{1}{\epsilon}$ and $t\parentheses*{\frac{1}{\epsilon}}:= 1$, respectively; proof details are given in Appendix~\ref{app:computational}. Then we strengthen Lemma~\ref{L3} by adding the regularity Assumption~\ref{A3} to get Lemma~\ref{L4}, and we prove the latter. Note that the proof of Theorem~\ref{T3} is exactly as for Theorem~\ref{T2}, but it uses Lemma~\ref{L4} rather than Lemma~\ref{L3}.

\begin{assumption}[Parameterized by $t\parentheses*{\frac{1}{\epsilon}}$]
\label{A2}
For every $\epsilon>0$ and every $M=\poly\parentheses*{t\parentheses*{\poly\parentheses*{\frac{1}{\epsilon}}}}$, one can compute in $\poly\parentheses*{t\parentheses*{\poly\parentheses*{\frac{1}{\epsilon}}}}$-time an explicit formula for an upper semi-continuous piecewise constant $u_{\epsilon,M}:\Delta\parentheses*{\Omega}\to \mathbb{R}_{\geq 0}$ s.t.:
     \begin{itemize}
     \item Every piece of $u_{\epsilon,M}$ covers a region of $\Delta\parentheses*{\Omega}$ which is a convex polygon with diameter at most~$\frac{\epsilon}{M}$.
     \item The total number of vertices of the above regions of $\Delta\parentheses*{\Omega}$ is $\poly\parentheses*{t\parentheses*{\poly\parentheses*{\frac{1}{\epsilon}}}}$.
     \item For every $q\in\Delta\parentheses*{\Omega}$ we have: $0\leq u_{\epsilon,M}\parentheses*{q}-u_{s}\parentheses*{q}\leq \epsilon$.
    \end{itemize}
\end{assumption}
  \begin{assumption}[Parameterized by $t\parentheses*{\frac{1}{\epsilon}}$]
  \label{A4}
  For every $1\leq i\leq m$, the $i$-th ex ante constraint is specified by $\f_i:\Delta\parentheses*{\Omega}\to \mathbb{R}$ s.t.~for every $\epsilon>0$, one can compute in $\poly\parentheses*{t\parentheses*{\poly\parentheses*{\frac{1}{\epsilon}}}}$-time an explicit formula for a $\poly\parentheses*{t\parentheses*{\poly\parentheses*{\frac{1}{\epsilon}}}}$-Lipschitz function $\g_i:\Delta\parentheses*{\Omega}\to \mathbb{R}$ s.t.~for every $q\in\Delta\parentheses*{\Omega}$:  
    \begin{equation*}
  0 \leq f_i\parentheses*{q}-g_i\parentheses*{q}\leq \epsilon.
  \end{equation*}
\end{assumption}
 \begin{lemma}
     \label{L3}
     Suppose that $k$ is constant, $u_s$ satisfies Assumption~\ref{A2} with $t\parentheses*{\frac{1}{\epsilon}}$ and we have $m$ ex ante constraints satisfying Assumption~\ref{A4} with $t\parentheses*{\frac{1}{\epsilon}}$. Then either the set of valid signaling schemes is empty or for every $\epsilon>0$, there exists a $\poly\parentheses*{m,t\parentheses*{\poly\parentheses*{\frac{1}{\epsilon}}}}$-time algorithm that computes an additively $\epsilon$-optimal signaling policy that violates each ex ante constraint at most by $\epsilon$.
     \end{lemma}
    The proof of Lemma~\ref{L3} first strengthens Assumption~\ref{A4} and assumes that the constraints are specified by $\poly\parentheses*{t\parentheses*{\poly\parentheses*{\frac{1}{\epsilon}}}}$-Lipschitz functions. Then we restrict ourselves to a grid consisting of the vertices of the pieces of $u_{\epsilon,M}$, where $M$ is the maximal Lipschitz constant among the functions specifying the constraints, and output the resultant optimal valid signaling scheme for $u_{\epsilon,M}$ rather than $u_s$. Finally, we estimate the loss in Sender's utility and the constraint values using the approximability guarantees.
    \begin{proof}[Proof of Lemma~\ref{L3}] 
    We strengthen Assumption~\ref{A4} to the following: the $i$-th ex ante constraint ($1\leq i\leq m$) is specified by a $\poly\parentheses*{t\parentheses*{\poly\parentheses*{\frac{1}{\epsilon}}}}$-Lipschitz function $\f_i:\Delta\parentheses*{\Omega}\to \mathbb{R}$ and some constant $c_i$. The original lemma follows from applying the lemma under the strengthened Assumption~\ref{A4} with $\epsilon$ replaced by $\frac{\epsilon}{2}$ and the $f_i$s replaced by the $g_i$s. This is because the original Assumption~\ref{A4} ensures that upon replacing $f_i$ with $g_i$, every valid signaling scheme remains such and $E\brackets*{\Sigma, f_i}$ decreases at most by $\epsilon$.
    
    Now we prove the lemma under the strengthened Assumption~\ref{A4}. Suppose that a valid signaling scheme exists and let $OPT$ be Sender's expected utility under an optimal valid scheme. Fix $\epsilon>0$ and let $M$ be the maximal Lipschitz constant among the $f_i$s. Compute an explicit formula for $u_{\epsilon,M}$. Let $q_1,...,q_n$ be the vertices of the regions of $\Delta\parentheses*{\Omega}$ covered by the pieces of $u_{\epsilon,M}$. Let us solve the following:
    
    \begin{align*}
    \text{max}\;\;\;\;&E\brackets*{\Sigma, u_{\epsilon,M}}\\
    \text{s.t.}\;\;\;\;\;\;&\p\brackets*{\omega_0}=E\brackets*{\Sigma, \p_{\sigma}\brackets*{\omega_0}}\;\; \forall \omega_0 \in\Omega\\
    &\supp\braces*{\Sigma}\subseteq\braces*{q_1,...,q_n}\\
    &E\brackets*{\Sigma, f_i}\leq c_i+\epsilon\;\;\forall 1\leq i\leq m
    \end{align*}
    
     This problem defines a finite LP with $n$ variables and $k+m$ constraints (as in Theorem~\ref{T1} proof, we should add a constraint for the probability masses in $\Sigma$ to sum up to $1$, but then we could remove one of the Bayes-plausibility constraints); this LP can be solved in time $\poly\parentheses*{n,k+m}=\poly\parentheses*{m,t\parentheses*{\poly\parentheses*{\frac{1}{\epsilon}}}}$. We return its solution $\Sigma$ as the desired signaling scheme.
     
    By the design of our LP, $\Sigma$ is Bayes-plausible and violates each ex ante constraint at most by $\epsilon$. Take now a valid optimal signaling scheme $\Sigma_{OPT}$ (for Sender's utility function $u_s$ rather than $u_{\epsilon,M}$). For every piece of $u_{\epsilon,M}$, move all the probability weight in $\Sigma$ from the region covered by this piece to the extreme points of that region in an expectation-preserving way (so Bayes-plausibility still holds) and denote the resultant signaling scheme by $\Sigma_{OPT}'$. Since the diameter of every such region is at most $\frac{\epsilon}{M}$ and the ex ante constraints have Lipschitz constants $\leq M$, we get that each ex ante constraint is violated at most by $\frac{\epsilon}{M}\cdot M=\epsilon$. Thus, $\Sigma_{OPT}'$ is a feasible solution to our LP, so $E\brackets*{\Sigma_{OPT}', u_{\epsilon,M}}\leq E\brackets*{\Sigma, u_{\epsilon,M}}$.
    
    Since $u_{\epsilon,M}$ is upper semi-continuous and piecewise constant we have:\\
    $E\brackets*{\Sigma_{OPT}, u_{\epsilon,M}}\leq E\brackets*{\Sigma_{OPT}', u_{\epsilon,M}}$. Furthermore, the third bullet from Assumption~\ref{A2} yields: $E\brackets*{\Sigma, u_{\epsilon,M}}-E\brackets*{\Sigma, u_{s}}\leq\epsilon$ and $E\brackets*{\Sigma_{OPT}, u_s}\leq E\brackets*{\Sigma_{OPT}, u_{\epsilon,M}}$. Combining the last four inequalities implies: $E\brackets*{\Sigma, u_s}\geq E\brackets*{\Sigma_{OPT}, u_s}-\epsilon=OPT-\epsilon$.
    \end{proof}

Now we formulate and prove Lemma~\ref{L4} -- a strengthening of Lemma~\ref{L3} used to prove Theorem~\ref{T3}.

     \begin{lemma}[Parameterized by $t\parentheses*{\frac{1}{\epsilon}}$]
     \label{L4}
     Suppose that $k$ is constant, $u_s$ satisfies Assumption~\ref{A2} with $t\parentheses*{\frac{1}{\epsilon}}$ and we have $m$ ex ante constraints satisfying Assumption~\ref{A4} with $t\parentheses*{\frac{1}{\epsilon}}$ and Assumption~\ref{A3}. Then for every $\epsilon>0$, there exists a $\poly\parentheses*{m,t\parentheses*{\poly\parentheses*{\frac{1}{\epsilon}}}}$-algorithm computing an additively $\epsilon$-optimal valid signaling policy.
     \end{lemma}
     The algorithm applies Lemma~\ref{L3} to a persuasion problem with strengthened ex ante constraints. The analysis compares the output to a convex combination of two outputs of Lemma~\ref{L3} -- one might violate the ex ante constraints and the other satisfies them with strict inequality. We use the proof of Lemma~\ref{L3} to bound the utility loss. 
      \begin{proof}[Proof of Lemma~\ref{L4}]
     $u_s$ is upper semi-continuous and defined on a compact set; thus, it is bounded from above by some constant $C$; assume w.l.o.g.~that $C>2$. Let $OPT$ be Sender's optimal utility for a valid scheme. Restrict ourselves to small enough values of $0<\epsilon<\frac{2}{C}$ s.t.~strengthening each ex ante constraint by $\epsilon$ leaves the set of valid signaling schemes nonempty (it is possible by Assumption~\ref{A3}).\footnote{To be precise, we assume that an upper bound on such values of $\epsilon$ is known in advance.} We return the signaling scheme $\Sigma$ outputted by the algorithm from Lemma~\ref{L3} on $0.5\epsilon$ and the problem obtained by strengthening each ex ante constraint by $0.5\epsilon$. Then $\Sigma$ satisfies the original constraints; it remains to bound its utility loss compared to $OPT$.
     
     Let $\Sigma'$ be the output of Lemma~\ref{L3} on $0.125\epsilon^3$ and the original problem; denote by $\Sigma''$ the output of Lemma~\ref{L3} on $0.5\epsilon$ and the problem obtained by straightening each original ex ante constraint by $\epsilon$.
     
     Let $M$ be the maximal Lipschitz constant among the $g_i$s from Assumption~\ref{A4}. Then $M$ is not affected by adding constant factors to the constraints; furthermore, note that by Assumption~\ref{A2}, $u_{0.125\epsilon^3,M}$ can also serve as $u_{0.5\epsilon,M}$ (since $0.125\epsilon^3<0.5\epsilon$ and $\frac{1}{0.125\epsilon^3}=\poly\parentheses*{\frac{1}{0.5\epsilon}}$). Therefore, by the proof of Lemma~\ref{L3}, we can assume w.l.o.g.~that $\Sigma,\Sigma',\Sigma''$ are all supported on the vertices of the pieces of $u_{0.125\epsilon^3,M}$; furthermore, $\frac{1}{1+0.25\epsilon^2}\Sigma'+\frac{0.25\epsilon^2}{1+0.25\epsilon^2}\Sigma''$ satisfies each original ex ante constraint. Note that $\Sigma$ is $0.5\epsilon$-additively-optimal among the schemes supported on the above extreme points and satisfying the original ex ante constraints, since $\Sigma$ is exactly optimal among such schemes if we replace $u_s$ with $u_{0.5\epsilon,M}$, by Lemma~\ref{L3} proof. Thus:
     \begin{align*}
         &E\brackets*{\Sigma,u_s}\geq E\brackets*{\frac{1}{1+0.25\epsilon^2}\Sigma'+\frac{0.25\epsilon^2}{1+0.25\epsilon^2}\Sigma'',u_s}-\frac{\epsilon}{2}=\frac{E\brackets*{\Sigma',u_s}}{1+0.25\epsilon^2}+\frac{0.25\epsilon^2 E\brackets*{\Sigma'',u_s}}{1+0.25\epsilon^2}-\frac{\epsilon}{2}\geq\\
         &\frac{OPT-0.125\epsilon^3}{1+0.25\epsilon^2}-\frac{\epsilon}{2}\geq OPT-\epsilon,
     \end{align*}
     where the last transition follows from $\frac{\epsilon}{2}<\frac{1}{C}\leq\frac{1}{OPT}$.
    \end{proof}

\section{Ex Ante vs.~Ex Post Constraints}
\label{sec:vs}
In this section, we bound the multiplicative gap in the Sender's optimal utility between ex ante constraints and the corresponding ex post constraints; we apply our bound to signaling in ad auctions in Subsection~\ref{sub:app}. 

In full generality, the gap can be arbitrarily large even for $k=2$ states of nature and $m=1$ convex constraints:

\begin{example}
\label{E1}
Fix $\epsilon\in\parentheses*{0,\frac{1}{2}}$; take $\Omega=\braces*{0,1}$ with a uniform prior; define $f\parentheses*{p_{\sigma}}:=p_{\sigma}\brackets*{\omega=1}$ and $c:=\frac{1}{2}+\epsilon$. Let $u_s\parentheses*{p_{\sigma}}$ be $0$ if
$p_{\sigma}\brackets*{\omega=1}\in\brackets*{0,\frac{1}{2}}$ and $2\cdot p_{\sigma}\brackets*{\omega=1}-1$ otherwise. The ex ante constraint specified by $f$ and $c$ allows full revelation, which yields expected utility of $\frac{1}{2}$ for Sender.

Convexity of $u_s$ implies that under the corresponding ex post constraint, there exists an optimal signaling scheme for which always $p_{\sigma}\brackets*{\omega=1}\in\braces*{0,c}$; straightforward calculations show that the Sender's optimal utility is $\frac{2\epsilon}{1+2\epsilon}$. Thus, the multiplicative gap tends to $\infty$ as $\epsilon$ tends to $0$.
\end{example}

We identify a multiplicatively-relaxed Jensen assumption on $u_s$ parameterized by $M\geq 1$, which combined with convexity of the $m$ constraints yields a multiplicative bound of $M^m$ on the gap between ex ante and ex post constraints.

\begin{assumption}[Parameterized by $M\geq 1$]
\label{A1}
For every $\lambda\in [0,1]$ and $p_{\sigma_1},p_{\sigma_2}\in\Delta\parentheses*{\Omega}$:
\begin{align*}
\lambda u_s\parentheses*{p_{\sigma_1}} +
(1-\lambda) u_s\parentheses*{p_{\sigma_2}}\leq M\cdot u_s\parentheses*{\lambda p_{\sigma_1} +(1-\lambda)p_{\sigma_2}}.
\end{align*}
\end{assumption}

For example, in Appendix~\ref{app:auctions}, we show that Assumption~\ref{A1} holds with $M=2$ for both the welfare and the revenue utility functions in the single-item, second-price auction setting. 
We note that there are utilities $u_s$ for which the assumption does not hold for any finite $M$: those $u_s$ that ``grow too slowly'' near $0$ (in particular, if $u_s$ maps a nonzero measure of the domain to $0$, as in Example~\ref{E1}).

 \begin{theorem}[A bound on the multiplicative gap between ex ante and ex post constraints]
    \label{T5}
    Suppose that $u_s$ satisfies Assumption~\ref{A1} with parameter $M\ge1$. Fix $m$ convex ex ante constraints and let $\Sigma_{\text{ex ante}}$ be a valid signaling scheme. Then there exists $\Sigma_{\text{ex post}}$, a valid signaling scheme under the corresponding $m$ ex post constraints, s.t.:\\
    $E\brackets*{\Sigma_{\text{ex post}}, u_s}\geq \frac{1}{M^m}\cdot E\brackets*{\Sigma_{\text{ex ante}}, u_s}.$
    \end{theorem}
    
    The proof runs Algorithm~\ref{alg:a1} for each constraint separately. This algorithm repeatedly pools a posterior violating the ex post constraint with a posterior satisfying this constraint with a strict inequality, replacing one of them by a posterior on which the ex post constraint is tight and decreasing the probability mass assigned to the other posterior. This process stops since each iteration decreases the number of posteriors in $\supp\parentheses*{\Sigma}$ on which the ex post constraint is not tight. The constraint convexity assures that the resultant scheme satisfies the ex post constraint; Assumption~\ref{A1} implies that the multiplicative loss caused by the pooling process (for each constraint) is at most $M$. Formally, we start with the following lemma.
\begin{algorithm}[tb]
\caption{Ex ante to ex post}
\label{alg:a1}
\textbf{Input}: A signaling scheme $\Sigma$ with a finite support satisfying: $E\brackets*{\Sigma, f}\leq c$.\\
\textbf{Parameters}: A continuous convex function $f:\Delta\parentheses*{\Omega}\to\mathbb{R}$, a constant $c$.\\
\textbf{Output}: An updated signaling scheme $\Sigma$ with a multiplicative expected utility loss of at most $M$ compared to the input s.t.~$\forall\p_{\sigma}\in\supp\parentheses*{\Sigma}:\;f\parentheses*{\p_{\sigma}}\leq c$.
\begin{algorithmic}[1]
\STATE $S\gets \supp\parentheses*{\Sigma}\cap f^{-1}\parentheses*{\parentheses*{-\infty,c}}$.
\STATE $T\gets \supp\parentheses*{\Sigma}\cap f^{-1}\parentheses*{\parentheses*{c,\infty}}$.
\WHILE{$S,T\neq \emptyset$}
\STATE Take $q_S\in S, q_T\in T$.
\STATE $r_S\gets\Pr_{p_{\sigma}\sim\Sigma}\brackets*{p_{\sigma}=q_S}$, $r_T\gets\Pr_{p_{\sigma}\sim\Sigma}\brackets*{p_{\sigma}=q_T}$.
\STATE Find $\lambda\in \parentheses*{0,1}$ s.t.~$f\parentheses*{\lambda q_S+\parentheses*{1-\lambda} q_T}=c$.
\STATE Define $q_c := \lambda q_S + \parentheses*{1-\lambda} q_T$.
\STATE $\supp\parentheses*{\Sigma}\gets\supp\parentheses*{\Sigma}\cup\braces*{q_c}$.
\IF {$\lambda r_T \geq \parentheses*{1-\lambda} r_S$}
\STATE $\supp\parentheses*{\Sigma}\gets\supp\parentheses*{\Sigma}\setminus\braces*{q_S}$.
\STATE $r_S\gets 0$, $r_T\gets r_T-\frac{\parentheses*{1-\lambda}r_S}{\lambda}, r_c\gets\frac{r_S}{\lambda}$.
\ELSE
\STATE $\supp\parentheses*{\Sigma}\gets\supp\parentheses*{\Sigma}\setminus\braces*{q_T}$.
\STATE $r_S\gets r_S-\frac{\lambda r_T}{1-\lambda}$, $r_T\gets 0, r_c\gets\frac{r_T}{1-\lambda}$.
\ENDIF
\STATE Update $\Sigma$ according to $r_S$, $r_T$, $r_c$.
\ENDWHILE
\STATE \textbf{return} $\Sigma$.
\end{algorithmic}
\end{algorithm}
\begin{lemma}
\label{L1}
Suppose that $u_s$ satisfies Assumption~\ref{A1} with some $M\geq 1$. Let $\Sigma_{\text{ex ante}}$ be a signaling scheme with a finite support satisfying a convex ex ante constraint specified by some $f$ and $c$. Then the output of Algorithm~\ref{alg:a1} on $\Sigma_{\text{ex ante}}$ is a signaling scheme $\Sigma_{\text{ex post}}$ satisfying the corresponding ex post constraint, s.t.: $E\brackets*{\Sigma_{\text{ex post}}, u_s}\geq \frac{1}{M}\cdot E\brackets*{\Sigma_{\text{ex ante}}, u_s}$.
\end{lemma}
Assuming Lemma~\ref{L1}, let us prove Theorem~\ref{T5}.
\begin{proof}[Proof of Theorem~\ref{T5}]
By Fact~\ref{F1}, assume w.l.o.g.~that $\Sigma_{\text{ex ante}}$ has a finite support. Let us run Algorithm~\ref{alg:a1} for $j=1,2,...,m$ on $\Sigma_{\text{ex ante}}$, $f_j$ and $c_j$ (updating the signaling scheme repeatedly) and let $\Sigma_{\text{ex post}}$ be the final output. Applying Lemma~\ref{L1} $m$ times, together with the convexity of the constraints -- which ensures that pooling cannot increase the expected value of any constraint function -- implies that $\Sigma_{\text{ex post}}$ satisfies the theorem requirements.
\end{proof}
It remains to prove the lemma.
\begin{proof}[Proof of Lemma~\ref{L1}]
Note that Algorithm~\ref{alg:a1} terminates after at most $\absolute*{\supp\parentheses*{\Sigma_{\text{ex ante}}}}-1$ iterations, since each iteration decreases $\absolute*{S\cup T}$, and throughout the algorithm run we have: $S\cup T\subseteq \supp\parentheses*{\Sigma_{\text{ex ante}}}$. The update rules ensure that $r_S$, $r_T$ and $r_c$ after each update are nonnegative and their sum equals $r_S+r_T$ before the update; furthermore, these updates preserve Bayes-plausibility, as $q_c=\lambda q_S + \parentheses*{1-\lambda} q_T$ and:
\begin{align*}
&0\cdot q_S + \parentheses*{r_T-\frac{\parentheses*{1-\lambda}r_S}{\lambda}} \cdot q_T +\parentheses*{\frac{r_S}{\lambda}} \cdot \parentheses*{\lambda q_S + \parentheses*{1-\lambda} q_T}=r_S\cdot q_S + r_T\cdot q_T,
\end{align*}
and also:
\begin{align*}
&\parentheses*{r_S-\frac{\lambda r_T}{1-\lambda}}\cdot q_S+ 0\cdot q_T +\frac{r_T}{1-\lambda} \cdot \parentheses*{\lambda q_S + \parentheses*{1-\lambda} q_T}=r_S\cdot q_S + r_T\cdot q_T.
\end{align*}
Therefore, $\Sigma$ remains a Bayes-plausible probability distribution throughout the algorithm run. In addition, the convexity of $f$ implies that the expectation of $f$ never increases. Hence, when the algorithm stops -- we must have $T=\emptyset$. Thus, $\Sigma_{\text{ex post}}$ satisfies the ex post constraint specified by $f$ and $c$. Moreover, Assumption~\ref{A1} implies that the multiplicative loss in the expected Sender's utility compared to $\Sigma_{\text{ex ante}}$ is at most $M$.
\end{proof}
    In Appendix~\ref{app:vs}, we prove the following facts on tightness of Theorem~\ref{T5} and our analysis. We leave as an open question the tightness of Theorem~\ref{T5} for general $m$.
    \begin{proposition}
    \label{P1}
    \begin{enumerate}
    \item Our analysis is tight for any $m$ and $M=2$.\footnote{Note that we use $M=2$ in our applications.}
    \item The bound from Theorem~\ref{T5} on the multiplicative gap between ex ante and ex post constraints is tight for $m=1$ and any $M$.
    \item This gap grows with $m$ and can be at least $m+1$.
    \end{enumerate}
    \end{proposition}

\subsection{Applications}
\label{sub:app}
We apply Theorem~\ref{T5} to the important domain of signaling in ad auctions. We use a generalization of the ``Bayesian Valuation Setting'' of~\citet{BBX} and add to it constraints on the signaling scheme.

Consider a single-item second-price auction with $n$ bidders. Recall from Section~\ref{sec:model} that the item being sold is the opportunity to show an online advertisement to a web user, whose characteristics are known to the auctioneer, but not to the bidders. Each bidder targets a certain set of users to whom showing her ad would be most valuable and the auctioneer signals information about which targeted sets the user belongs to.

In the language of persuasion, Sender is the auctioneer while Receiver is the set of bidders. Take $\Omega:=\braces*{0,1}^n$, where the $i$-th coordinate specifies whether the web user is in the $i$-th advertiser's targeted set; denote by $\omega=\parentheses*{\omega_1,...,\omega_n}$ the state of nature; let $\p$ be some commonly-known prior distribution. Assume further that  for every $1\leq i\leq n$, the $i$-th bidder has a private \emph{type} $t_i$; for every $1\leq i\leq n$, the \emph{valuation} $v_i$ of the $i$-th bidder is determined by a nonnegative function $v_i\parentheses*{\omega_i, t_i}$.\footnote{Unlike~\citet{BBX}, we assume neither that the $t_i$s are i.i.d. nor that $v_1\equiv...\equiv v_n$.} Fix $2n$ continuously differentiable CDFs $\LL_i$ and $\HH_i$ ($1\leq i\leq n$) and assume that $v_i\parentheses*{0,t_i}\sim_{t_i}\LL_i$ and $v_i\parentheses*{1,t_i}\sim_{t_i}\HH_i$ for every $1\leq i\leq n$. The auction runs as follows:
\begin{enumerate}
\item The auctioneer commits to a valid signaling scheme $\Sigma$, an allocation rule and a payment rule.
\item The auctioneer discovers the state of nature $\omega\in\Omega$.
\item The auctioneer broadcasts a public signal realization $\sigma$ according to $\Sigma\parentheses*{\omega}$.
\item The bidders update their expected valuations using $p_{\sigma}$ and report their bids to the auctioneer.
\item The auction outcome is determined by the allocation and the payment rules.
\end{enumerate}
Define Sender's utility $u_s\parentheses*{p_{\sigma}}$ to be the expected welfare -- the winner's value -- over $t_1,...,t_n$, for a posterior $p_{\sigma}$. Explicitly:
\begin{align*}
&u_s\parentheses*{p_{\sigma}}:=\mathbb{E}_{t_1,...,t_n}\brackets{\max\braces{p_{\sigma}\brackets*{\omega_1=0}\cdot v_1\parentheses*{0,t_1}+p_{\sigma}\brackets*{\omega_1=1}\cdot v_1\parentheses*{1,t_1},...,\\
&p_{\sigma}\brackets*{\omega_n=0}\cdot v_n\parentheses*{0,t_n}+p_{\sigma}\brackets*{\omega_n=1}\cdot v_n\parentheses*{1,t_n}}}.
\end{align*}
In Appendix~\ref{app:auctions}, we prove the following result.
\begin{proposition}
\label{P2}
$u_s$ -- the expected (over the bidders' private types) welfare in a single-item second-price auction with signaling -- satisfies Assumption~\ref{A1} with $M=2$.
\end{proposition}
This result extends to expected revenue and to sponsored search (slot) auctions -- see Appendix~\ref{app:auctions}. Proposition~\ref{P2} suggests the following ``recipe'' for solving signaling problems in ad auctions under a constant number of convex ex ante constraints: (approximately) solve the problem for the corresponding ex post constraints; this guarantees, by Theorem~\ref{T5}, a constant-factor approximation for the original problem. The next example demonstrates.
\begin{example}
\label{E2}
Take a single ex ante constraint specified by the function $(-\min\braces*{b_{\omega'} p_{\sigma}\brackets*{\omega=\omega'}}_{\omega'\in\Omega})$ with some constant weights $\braces*{b_{\omega'}}_{\omega'\in\Omega}$. As mentioned in Section~\ref{sec:model}, this constraint is a possible model for anti-discrimination. Finding the optimal valid scheme $\Sigma^*_{\text{ex ante}}$ is an open question. However, the corresponding ex post constraint is simple to handle -- it restricts the posteriors to an appropriate simplex, and since $u_s$ (the social welfare) is convex, the optimal scheme $\Sigma^*_{\text{ex post}}$ is supported precisely on the vertices of this simplex, and is uniquely specified by Bayes-plausibility. Theorem~\ref{T5}, combined with Proposition~\ref{P2}, shows that $\Sigma^*_{\text{ex post}}$ is a $\frac{1}{2}$-approximation to $\Sigma^*_{\text{ex ante}}$.
\end{example}

\section{Future Work}
\label{sec:future}
We study the setting of ex ante- and ex post-constrained persuasion, which has applications to areas including ad auctions and limited attention. A future research direction, especially considering Theorem~\ref{T5}, is studying (nearly) optimal signaling schemes under common ex post constraints, such as KL divergence. Another interesting direction is constrained persuasion with \emph{private} signaling, e.g., when Sender's utility is a function of the set of Receivers who adopt a certain action~\cite{AB}.
\bibliographystyle{plainnat}
\bibliography{main}
\appendix
\section{Alternative Proof of Fact~\ref{F1}}
\label{app:existence_alternative}
The proof is similar to that of Theorem~\ref{T1}, but we add to the infinite LP $m$ constraints corresponding to the given ex ante constraints, rather than restricting the support of $\Sigma$ to a compact subset of $\Delta\parentheses*{\Omega}$.
    \begin{proof}
    Denote the function and the constant specifying the $i$-th ex ante constraint ($1\leq i\leq m$) by $\f_i$ and $c_i$, respectively. We aim to solve:
    \begin{align*}
    \text{max}\;\;\;\;&E\brackets*{\Sigma,u_s}\\
    \text{s.t.}\;\;\;\;\;\;&\p\brackets*{\omega_0}=E\brackets*{\Sigma,\p_{\sigma}\brackets*{\omega_0}}\;\;\forall \omega_0 \in\Omega\\
    &E\brackets*{\Sigma,f_i}\leq c_i\;\;\forall 1\leq i\leq m
    \end{align*}
    The optimization problem specifies an infinite LP with $k+m$ constraints s.t.~the ``variables'' are the distribution $\Sigma$ over elements of $\Delta\parentheses*{\Omega}$. The rest of the proof is the same as for Theorem~\ref{T1} with $K=\Delta\parentheses*{\Omega}$ (using the notions of Theorem~\ref{T1} proof); the only difference is that we have $k+m$ linear constraints rather than $k$, which yields a bound of $2^{k+m}$ on $\absolute{\supp\parentheses*{\Sigma}}$ in the third step and a bound of $k+m$ in the fourth step.
   \end{proof}
\section{Proof of Proposition~\ref{P3}}
\label{app:existence}
\begin{proof}
Fix $k$, $m$ and some $\Omega$ of size $k$. We shall define $m$ ex ante constraints and an upper semi-continuous $u_s$ s.t.~any optimal signaling scheme has support size exactly $k+m$.

Let $e_1,...,e_k$ be the standard basis of $\mathbb{R}^k$. Take $q_1,...,q_m$ to be $m$ distinct interior points of $\Delta\parentheses*{\Omega}\setminus\braces*{e_1,...,e_k}$ with $\frac{q_1+...+q_m}{m}=\parentheses*{\frac{1}{k},...,\frac{1}{k}}$; set $u_s$ to be $1$ on $\braces*{q_1,...,q_m}$, $\frac{1}{2}$ on $\braces*{e_1,...,e_k}$ and $0$ on $\Delta\parentheses*{\Omega}\setminus\braces*{e_1,...,e_k, q_1,...,q_m}$; let $f_i$ ($1\leq i\leq m$) be some nonnegative continuous function, which is $1$ on $q_i$ and $0$ on the other $q_j$s; set $c_1=...=c_m=\frac{1}{2m}$; choose $p=\parentheses*{\frac{1}{k},...,\frac{1}{k}}$.

Then $u_s$ is upper semi-continuous and the $f_i$s are continuous. Furthermore, no valid signaling scheme under the ex ante constraints specified by the $f_i$s and the $c_i$s assigns probability greater than $\frac{1}{2m}$ to each $q_i$; thus, the utility of a valid scheme is at most: $m\cdot\frac{1}{2m}\cdot 1+\parentheses*{1-m\cdot\frac{1}{2m}}\cdot\frac{1}{2}=\frac{3}{4}$.

A utility of exactly $\frac{3}{4}$ is achieved by valid schemes that assign probability of exactly $\frac{1}{2m}$ to each $q_i$ and split the remaining probability between $e_1,...,e_k$. Bayes-plausibility implies that exactly one such scheme exists, assigning probability of $\frac{1}{2m}$ to each $q_i$ and probability of $\frac{1}{2k}$ to each $e_j$. Thus, there exists a single optimal valid scheme, which has support size exactly $k+m$.
\end{proof}
\begin{observation}
\label{O2}
One can modify our construction to make $u_s$ continuous.
\end{observation}
For example, one can make $u_s$ to quickly decrease to $0$ on all the rays originating from any $q_i$ or $e_j$ and require $f_i$ to be greater than $1$ in a deleted neighbourhood of $q_i$ on which $u_s$ is nonzero.
\section{Proof of Theorem~\ref{T2}}
\label{app:computational}
\begin{proof}[Proof of Theorem~\ref{T2}, part $1$]
Suppose that $u_s$ is either $O(1)$-Lipschitz or piecewise constant, having a constant pieces number, with each piece covering a convex polygon in $\Delta\parentheses*{\Omega}$ having a constant number of vertices. Then $u_s$ satisfies Assumption~\ref{A2} with $t\parentheses*{\frac{1}{\epsilon}}:=\frac{1}{\epsilon}$. Indeed, to define $u_{\epsilon,M}$, one can divide $\Delta\parentheses*{\Omega}$ to $\poly\parentheses*{\frac{1}{\epsilon}}$ simplices of diameters at most $\frac{\epsilon}{M}$ s.t.~the supremum and the infimum of $u_s$ on every simplex differ at most by $\epsilon$, and then set $u_{\epsilon,M}$ on every such simplex to be the supremum of $u_s$ on it.

An ex ante constraint specified by an $O(1)$-Lipschitz function trivially satisfies Assumption~\ref{A4} with $t\parentheses*{\frac{1}{\epsilon}}:=\frac{1}{\epsilon}$ -- simply define $g_i:=f_i$. Consider now an ex ante constraint specified by a function of the form:
\begin{equation*}
f_i\parentheses*{p_{\sigma}}:=b\cdot \sum_{1\leq j\leq l} \parentheses*{\sum_{\omega'\in\Omega_j} p_{\sigma}\brackets*{\omega'}}\ln\frac{\sum_{\omega'\in\Omega_j} p_{\sigma}\brackets*{\omega'}}{b_j},
\end{equation*}
where $\braces*{\Omega_j}_{1\leq j\leq l}$ is a partition of $\Omega$, $b_1,...,b_l>0$ are constants and $b$ is constant. We shall show that one can assume w.l.o.g.~that $f_i\parentheses*{p_{\sigma}}\equiv \sum_{\omega'\in\Omega} p_{\sigma}\brackets*{\omega'}\ln p_{\sigma}\brackets*{\omega'}$; then we shall prove that the corresponding ex ante constraint satisfies Assumption~\ref{A4} with $t\parentheses*{\frac{1}{\epsilon}}:=\frac{1}{\epsilon}$.

First, assume w.l.o.g.~that $b=1$: by dividing both $f_i$ and $c_i$ by $\absolute{b}$ we can assume $b\in\{-1,1\}$;\footnote{The case $b=0$ is trivial.} then note that if $g_i$ fits for $f_i$, then $-\epsilon-g_i$ fits for $-f_i$. Secondly, assume w.l.o.g.~that $l=k$ and $\absolute*{\Omega_j}=1$ for every $1\leq j\leq l$ -- just replace $\Omega$ with $\Omega':=\braces*{\Omega_j}_{1\leq j\leq l}$. Thirdly, by adding a linear function to $f_i$, assume w.l.o.g.~that $f_i\parentheses*{p_{\sigma}}\equiv\sum_{\omega'\in\Omega} p_{\sigma}\brackets*{\omega'}\ln p_{\sigma}\brackets*{\omega'}$; it is possible since adding an $O(1)$-Lipschitz function does not affect the satisfaction of Assumption~\ref{A4}.

Let $p'$ be the center of $\Delta\parentheses*{\Omega}$ and let $S_{\epsilon}$ be the contraction (homothety) of $\Delta\parentheses*{\Omega}$ centered at $p'$ with coefficient $\frac{1}{1+\epsilon^2}$. The restriction of $f_i$ to $S_{\epsilon}$ is $\poly\parentheses*{\frac{1}{\epsilon}}$-Lipschitz, since on $S_{\epsilon}$ one has: $||\bigtriangledown f_i\parentheses*{p_{\sigma}}||=O\parentheses*{\sum_{\omega'\in\Omega} \absolute*{\ln p_{\sigma}\brackets*{\omega'}}}=O\parentheses*{\sum_{\omega'\in\Omega}\frac{1}{p_{\sigma}\brackets*{\omega'}}}=\poly\parentheses*{\frac{1}{\epsilon}}$. Extend this restriction of $f_i$ to a function $\tilde{g_i}:\Delta\parentheses*{\Omega}\to\mathbb{R}$ s.t. for every $q\in\Delta\parentheses*{\Omega}\setminus S_{\epsilon}$, $\tilde{g_i}(q)$ equals the value of $f_i$ on the projection of $q$ onto the closed, convex and nonempty set $S_{\epsilon}$. Finally, set $g_i\equiv\tilde{g_i}-\frac{\epsilon}{2}$.

Then $g_i$ is $\poly\parentheses*{\frac{1}{\epsilon}}$-Lipschitz, since $f_i$ is $\poly\parentheses*{\frac{1}{\epsilon}}$-Lipschitz and projection on a closed, convex and nonempty set is $1$-Lipschitz. It remains to check that $0\leq f_i(q)-g_i(q)\leq\epsilon$ for every $q\in\Delta\parentheses*{\Omega}$. It is immediate for $q\in S_{\epsilon}$. Fix now $q\in \Delta\parentheses*{\Omega}\setminus S_{\epsilon}$. It is enough to show that $\absolute*{f_i(q)-\tilde{g_i}(q)}\leq\frac{\epsilon}{2}$. Indeed, by the definition of $\tilde{g_i}$, $\tilde{g_i}(q)=f_i\parentheses*{q'}$, where $q'$ is the projection of $q$ onto $S_{\epsilon}$. By the choice of $S_{\epsilon}$ we have $||q-q'||\leq\epsilon^2$. Therefore, the change in $f_i$ between $q$ and $q'$ is at most (for small enough $\epsilon$):
\begin{equation*}
    k\absolute*{\epsilon^2\ln\epsilon^2}=O\parentheses*{\epsilon^2\absolute*{\ln\epsilon}}=o\parentheses*{\epsilon},
\end{equation*}
so $\absolute*{f_i(q)-\tilde{g_i}(q)}=\absolute*{f_i(q)-f_i\parentheses*{q'}}=o\parentheses*{\epsilon}$. Thus, $f_i$ indeed satisfies Assumption~\ref{A4} with $t\parentheses*{\frac{1}{\epsilon}}:=\frac{1}{\epsilon}$, as desired.

We proved that $u_s$ and the constraints satisfy Assumptions~\ref{A2} and~\ref{A4} with $t\parentheses*{\frac{1}{\epsilon}}:=\frac{1}{\epsilon}$; thus, Theorem~\ref{T2}, part $1$ follows from Lemma~\ref{L3}.
\end{proof}
\begin{proof}[Proof of Theorem~\ref{T2}, part $2$]
Fix a constant $\epsilon>0$. If $u_s$ is piecewise constant, with a constant number of pieces, s.t.~each piece covers a convex polygon in $\Delta\parentheses*{\Omega}$ having a constant vertex number -- it satisfies Assumption~\ref{A2} with $t\parentheses*{\frac{1}{\epsilon}}:= 1$: to define $u_{\epsilon,M}$, one can just refine the pieces of $u_s$ by division to simplices of diameters at most $\frac{\epsilon}{M}$. If $u_s$ is continuous, then from the compactness of $\Delta\parentheses*{\Omega}$ and the Heine–Cantor theorem, we get that $u_s$ is uniformly continuous. Therefore, $u_s$ satisfies Assumption~\ref{A2} with $t\parentheses*{\frac{1}{\epsilon}}:= 1$: to define $u_{\epsilon,M}$, one can divide $\Delta\parentheses*{\Omega}$ to simplices of small enough diameters; then one should define $u_{\epsilon,M}$ on every such simplex to be the supremum of $u_s$ on~it.

Furthermore, each ex ante constraint satisfies Assumption~\ref{A4} with $t\parentheses*{\frac{1}{\epsilon}}:= 1$. Indeed, given a continuous $f_i$, the compactness of $\Delta\parentheses*{\Omega}$ and the Heine–Cantor theorem implies that $f_i$ is uniformly continuous. To define $g_i$, one should divide $\Delta\parentheses*{\Omega}$ to simplices of small enough diameter; then one should temporarily set $g_i\equiv f_i$ on the vertices of the simplices and extend $g_i$ linearly on each simplex; finally, one should slightly shift down $g_i$ so that it is never above $f_i$.

Therefore, $u_s$ and the constraints satisfy Assumptions~\ref{A2} and~\ref{A4} with $t\parentheses*{\frac{1}{\epsilon}}:= 1$; hence, Theorem~\ref{T2}, part $2$ follows from Lemma~\ref{L3}.
\end{proof}
\section{Proposition~\ref{P2} -- Proof and Similar Results}
\label{app:auctions}
In this appendix, we formulate and prove the technical Lemma~\ref{L2}; we use it to prove Proposition~\ref{P2} and to demonstrate analogous results for other auction settings, including revenue maximization in single-item, second-price auctions and welfare maximization in sponsored search (slot) auctions.
\begin{lemma}
\label{L2}
Fix $1\leq j\leq n$ and $n$ linear functions $g_1,...,g_n: \Delta\parentheses*{\Omega}\to\mathbb{R}_{\geq 0}$. Define $g^j:\Delta\parentheses*{\Omega}\to\mathbb{R}_{\geq 0}$ by setting $\g^{j}\parentheses*{y}$, for every $y\in \Delta\parentheses*{\Omega}$, to be the $j$-th maximal number among $g_1\parentheses*{y},...,g_n\parentheses*{y}$. Then $g^j$ satisfies Assumption~\ref{A1} with $M=2$.
\end{lemma}
\begin{proof}
Fix $y,z\in \Delta\parentheses*{\Omega}$ and $\lambda\in[0,1]$. We have to prove:
\begin{equation*}
\lambda g^j\parentheses*{y} +
(1-\lambda) g^j\parentheses*{z}\leq 2 g^j\parentheses*{\lambda y +(1-\lambda) z}.
\end{equation*}
Assume w.l.o.g.~that $\lambda g^j\parentheses*{y}\geq (1-\lambda) g^j\parentheses*{z}$ and that $g_1\parentheses*{y}\geq g_2\parentheses*{y}\geq...\geq g_n\parentheses*{y}$. Then we get:
\begin{align*}
&\lambda g^j\parentheses*{y} +(1-\lambda) g^j\parentheses*{z}\leq 2\lambda g^j\parentheses*{y}=2\lambda g_j (y)=2\lambda\min_{1\leq i\leq j}\braces*{g_i (y)}\leq 2\min_{1\leq i\leq j}\braces*{\lambda g_i (y)+ (1-\lambda)g_i (z)}=_{\parentheses*{*}}\\
&2\min_{1\leq i\leq j}\braces*{g_i\parentheses*{\lambda y +(1-\lambda) z}}\leq2 g^j\parentheses*{\lambda y +(1-\lambda) z},
\end{align*}
where $\parentheses*{*}$ follows from the linearity of the $g_i$s.
\end{proof}
Now we prove -- using Lemma~\ref{L2} -- Proposition~\ref{P2} and discuss analogous results for other auction settings.
\begin{proof}[Proof of Proposition~\ref{P2}]
Note that $u_s$ is an expectation of maximum of linear nonnegative functions. By Lemma~\ref{L2}, every term in the expectation satisfies Assumption~\ref{A1} with $M=2$; thus, the same holds for the expectation.
\end{proof}
{\bf Similar results.} The applications of Theorem~\ref{T5} go beyond the setting from Section~\ref{sec:vs}, which was chosen for the sake of simplicity. In particular, the revenue in single-item, second-price auctions and the welfare in sponsored search (slot) auctions are also expectations of certain linear combinations of functions of the form $g^j$ as described in Lemma~\ref{L2}. Therefore, they too satisfy Assumption~\ref{A1} with $M=2$.
\section{Proof of Proposition~\ref{P1}}
\label{app:vs}
\begin{proof}
\begin{enumerate}
\item Fix $m$ and $k=2^m$. We shall show that the analysis from Section~\ref{sec:vs} is tight for $M=2$.

Take $\Omega:=\braces*{0,1}^m$ and $u_s\parentheses*{p_{\sigma}}:=||p_{\sigma}||_{\infty}$; let $p$ be uniform over $\Omega$. By Lemma~\ref{L2} (with $j=1$ and $g_{\omega_0}:=p_{\sigma}\brackets*{\omega=\omega_0}$ for every $\omega_0\in\Omega$), $u_s$ satisfies Assumption~\ref{A1} with $M=2$. Define for every $1\leq i\leq m$: $f_i\parentheses*{p_{\sigma}}:=p_{\sigma}\brackets*{\omega_i=1}=\sum_{\omega'=\parentheses*{\omega'_1,...,\omega'_m}\in\Omega: \omega'_i=1} p_{\sigma}\brackets*{\omega=\omega'}$ and $c_i:=\frac{1}{2}$. For every $p_{\sigma}\in\Delta\parentheses*{\Omega}$, denote by $R\brackets*{p_{\sigma},i}$ the posterior obtained from $p_{\sigma}$ by assigning to each $\omega'\in\Omega$ the probability assigned by $p_{\sigma}$ to the state of nature obtained from $\omega'$ by reversing its $i$-th bit.

Consider the following $m$ runs of Algorithm~\ref{alg:a1} for the $m$ constraints. Start with $\Sigma$ representing full revelation (which is valid under the ex ante constraints specified by the $f_i$s and the $c_i$s). Then on the $i$-th run of Algorithm~\ref{alg:a1}, pool every posterior $p_{\sigma}$ with $R\brackets*{p_{\sigma},i}$.

Inductively, just before two posteriors are pooled together, they have equal probability weights in the signaling scheme; therefore, their probability weights are moved entirely to the new posterior that the pooling creates. Note that for every $1\leq i\leq m$, the $i$-th run of Algorithm~\ref{alg:a1} is legal, since it only pools posteriors having $f_i=0$ with posteriors having $f_i=1$; furthermore, at the end of the $i$-th run, all the posteriors in $\supp\parentheses*{\Sigma}$ have $f_i=\frac{1}{2}$. Moreover, inductively, at the end of the $i$-th run, every posterior in $\supp\parentheses*{\Sigma}$ specifies deterministically the last $m-i$ bits of $\omega$ and induces a uniform distribution on $\braces*{0,1}^i$ for the prefix of length $i$ of $\omega$.

Therefore, after the $m$-th run, we end with $\supp\parentheses*{\Sigma}=\{p\}$ (i.e., the no revelation policy), yielding expected Sender's utility of $\frac{1}{k}$. We started with the full revelation policy, yielding utility of $1$; thus, the total multiplicative utility loss is $k=2^m=M^m$.

\item Assume that $m=1$ and fix $M\geq 1$. We shall define $u_s$ satisfying Assumption~\ref{A1} with $M$ and an ex ante constraint outperforming the corresponding ex post constraint by a multiplicative factor of $M$.

Take: $\Omega:=\braces*{0,1}$; $p$ uniform over $\Omega$; $f:=p_{\sigma}\brackets*{\omega=1}$; $c:=\frac{1}{2}$; and $u_s\parentheses*{p_{\sigma}}:=\frac{1}{M}+\\\absolute*{p_{\sigma}\brackets*{\omega=1}-\frac{1}{2}}\cdot \frac{2\parentheses*{M-1}}{M}$.

Then $u_s\parentheses*{p_{\sigma}}\in\brackets*{\frac{1}{M},1}$ for every $p_{\sigma}\in\Delta\parentheses*{\Omega}$; thus, $u_s$ satisfies Assumption~\ref{A1} with $M$. Furthermore, $f$ is linear, thus convex. Under the ex post constraint specified by $f$ and $c$, the only valid signaling scheme has support $\{p\}$; hence, the optimal expected Sender's utility is $\frac{1}{M}$; the corresponding ex ante constraint allows full revelation, which yields expected Sender's utility of $1$. Therefore, we have a multiplicative gap of $M$ between the two constraint types.

\item Fix $m$ and $k=m+1$. We shall prove that the multiplicative gap between ex ante and ex post constraints can be $m+1$ for $M=2$.

Take $\Omega:=\braces*{1,...,k}$ with $p$ uniform on $\Omega$; set $u_s\parentheses*{p_{\sigma}}:=||p_{\sigma}||_{\infty}$, $f_i:=p_{\sigma}\brackets*{\omega=i}$ and $c_i:=\frac{1}{k}$ ($1\leq i\leq m$). As explained in our proof of part $1$, $u_s$ satisfies Assumption~\ref{A1} with $M=2$.

On the one hand, for the ex post constraints specified by the $f_i$s and the $c_i$s, the only valid signaling scheme has support $\{p\}$, yielding expected Sender's utility of $\frac{1}{k}$. On the other hand, the corresponding ex ante constraints allow full revelation, yielding utility of 1. Thus, we get a multiplicative gap of $k=m+1$.
\end{enumerate}
\end{proof}
\end{document}